\documentclass[11pt]{article}
\usepackage[margin=1in]{geometry}
\usepackage{amsmath,amssymb,amsbsy,amsgen,amsopn,amsthm}
\usepackage{graphicx}
\usepackage{mathrsfs}

\newcommand{\Rset}{\mathbb{R}}

\newcommand{\Ffam}{\mathcal{F}}
\newcommand{\Mfam}{\mathcal{M}}

\newcommand{\Afam}{\mathcal{A}}
\newcommand{\Dfam}{\mathcal{D}}
\newcommand{\Lfam}{\mathcal{L}}

\newcommand{\Vfam}{\mathcal{V}}
\newcommand{\opt}{\mathrm{OPT}}

\usepackage{algorithm,algorithmic}

\newtheorem{theorem}{Theorem}
\newtheorem{lemma}{Lemma}
\newtheorem{corollary}{Corollary}
\newtheorem{definition}{Definition}

\begin{document}

\title{Approximation algorithms for highly
connected \\ multi-dominating sets in unit disk graphs}

\author{Takuro Fukunaga\thanks{This work was done when the authors
was at National Institute of Informatics, Japan. Email: {\tt takuro.fukunaga@riken.jp}}}

\date{}

\maketitle

\begin{abstract}
  Given an undirected graph on a node set $V$
  and positive integers $k$ and $m$,
  a $k$-connected $m$-dominating set ($(k,m)$-CDS)
  is defined as a subset $S$ of $V$ 
  such that each node in $V \setminus S$
  has at least $m$ neighbors in $S$, and
 a $k$-connected subgraph is induced by $S$.
  The weighted $(k,m)$-CDS problem is to
  find a minimum weight $(k,m)$-CDS
  in a given node-weighted graph.
  The problem is called
  the unweighted $(k,m)$-CDS problem
  if the objective
  is to minimize the cardinality of a $(k,m)$-CDS.
  These problems
  have been actively studied for unit disk graphs,
  motivated by the application of constructing a virtual backbone 
  in a wireless ad hoc network.
  In this paper,
  we consider the case in which $k \leq m$,
  and we present a simple
  $O(k5^k)$-approximation algorithm for the unweighted $(k,m)$-CDS
 problem,
 and a primal-dual $O(k^2 \log k)$-approximation algorithm for the
 weighted $(k,m)$-CDS problem. 
 \end{abstract}

\section{Introduction}
\label{sec.intro}

A wireless ad hoc network is a decentralized wireless network.
Compared with a traditional communication
network, it has the advantage of not requiring
any infrastructure, such as base stations and WiFi routers; this is a
great benefit when operating sensor networks, vehicle networks, or
networks in disaster areas.  However, for efficient operation of a
wireless ad hoc network, we have to overcome many technical challenges.
One of these is to reduce the redundant communication caused by flooding
messages.  
A typical solution 
for this task
is to construct a
\emph{virtual backbone network}, as follows:
we choose several backbone nodes from the network,
and then we construct a subnetwork that comprises only these backbone nodes.
When a message arrives, first, it is delivered to a backbone node,
next, it is flooded to all the backbone nodes via the virtual backbone
network,
and finally, the destination node receives the message from an adjacent backbone node.
This reduces redundancy more as the virtual backbone network is
smaller.
However, it is also important that the virtual backbone network
is fault tolerant.

Developing algorithms for constructing a virtual backbone network
is an active area of research.
A promising approach
is to formulate
a virtual backbone network 
as a \emph{connected dominating set} (CDS), and
to consider an algorithm for finding a minimum cardinality or a minimum
weight CDS.
For an undirected graph with a node set $V$, a CDS
is defined as a subset $S$ of $V$ such that
each node in $V \setminus S$ is adjacent to at least one node in $S$,
and $S$ induces a connected subgraph.
This approach has gained in popularity,
and many papers have been published~~\cite{AmbuhlEMN06,ChengHLWD03,GuhaK98,GuhaK99,MaratheBHRR95,ZouLGW09}.
It is typically assumed that the input graph is a unit disk graph,
which is a natural choice for modeling a wireless network.
Since the problem of finding the minimum
cardinality CDS
is NP-hard even for unit disk graphs~\cite{ClarkCJ90},
some studies have considered approximation algorithms.

A CDS does not give a fault-tolerant virtual backbone network.
This is because a CDS is only required to be connected,
and each node outside a CDS is required to have only one neighbor in the CDS.
Hence, if a backbone node fails,
the virtual backbone network may be disconnected,
or a non-backbone node may lose access to the virtual backbone network.
To overcome this disadvantage, 
Dai and
Wu~\cite{Dai:2006}
proposed replacing a CDS by a \emph{$k$-connected $k$-dominating set},
and they addressed the problem of finding a minimum $k$-connected
$k$-dominating set in a unit disk graph.
For a graph with the node set $V$,
a subset $S$ of $V$
is called \emph{$k$-connected} if the subgraph induced by $S$
is $k$-connected (i.e., it is connected even if any $k-1$ nodes are
removed),
and is called \emph{$k$-dominating} if each node $v \in V\setminus S$ has $k$
neighbors in $S$.
Triggered by their study, 
much attention has been paid to this problem,
extending the notion of 
a $k$-connected $k$-dominating set to a more-general $k$-connected $m$-dominating
set ($(k,m)$-CDS).

The problem of finding a minimum cardinality $(k,m)$-CDS in a unit disk
graph is called the \emph{unweighted $(k,m)$-CDS problem}.
If each node is given
a nonnegative weight, and
the objective is to minimize the weight of a $(k,m)$-CDS,
then this is called the \emph{weighted $(k,m)$-CDS problem}.
As for the unweighted $(k,m)$-CDS problem,
several 
constant-approximation algorithms were given
for $k\leq 3$~\cite{Shang:2007jg,ShiZZW16,WangKAGLZW13,Wang:2015,Zhang16}.
As for the weighted $(k,m)$-CDS problem,
there are several constant-approximation algorithms for $k=m=1$~\cite{AmbuhlEMN06,ZouLGW09},
but no approximation algorithm was known for
the case of $(k,m)\neq (1,1)$ before our study
(see Section~\ref{sec.relatedworks} for more literature reviews).

After these previous studies,
a natural question arises as to whether there
is a constant-approximation algorithm
for the unweighted $(k,m)$-CDS problem with $k \geq 4$,
and for the weighted problem with $(k,m)\neq (1,1)$.
For the unweighted problem, this question has been already addressed in
both \cite{WangKAGLZW13} and \cite{Wang:2015}.

\subsection{Our results}
\label{sec.ourresults}
We answer the above question affirmatively.
The main contribution of this paper is to present
constant-approximation algorithms
for both the unweighted and the weighted $(k,m)$-CDS problems
when $k$ is a constant and $k \leq m$.

Specifically, we present two algorithms;
one is an $O(k 5^k)$-approximation algorithm for the unweighted $(k,m)$-CDS problem,
and the other is an $O(k^2 \log k)$-approximation algorithm for the
weighted $(k,m)$-CDS problem.
Notice that both algorithms achieve a constant factor if $k$ is a constant.
The approximation factor of the latter algorithm is better than that of
the former,
and it can be applied to the weighted problem, while the former algorithm is
restricted to the unweighted problem.
However, the former algorithm is simple, easy to analyze,
and can also be applied to other graph classes.
In fact,
 for $k\in \{2,3\}$,
 the former algorithm is obtained by introducing more
 specification into the algorithms
 given in~\cite{Shang:2007jg,Wang:2015}.
 Hence our analysis on the former algorithm gives a simple proof for the
 fact that the algorithms in~\cite{Shang:2007jg,Wang:2015} achieve
 constant-approximation for $k\in \{2,3\}$.

When the author published a preliminary version \cite{Fukunaga15arxiv} of this work,
he noticed that
Shi, Zhang, and Du~\cite{ShiZD15} also presented an $O(k^2 \log k)$-approximation
algorithm for the weighted $(k,m)$-CDS problem in unit disk graphs with
$k\leq m$.
The journal version of their paper was published as Shi~et~al.\ \cite{ShiZMD17}.
Their study has been done independently from ours,
but the taken approach is similar to ours
in the fact that
both of their and our algorithms require
computing an $m$-dominating set,
and then convert it into a $k$-connected $m$-dominating set.
In the first step,
both of the algorithms require
an approximation algorithm for the minimum weight $m$-dominating
set problem.
Shi, Zhang, and Du~\cite{ShiZD15}  suggested using the algorithm
given in~\cite{WillsonZWD15} for this purpose, but
the algorithm of~\cite{WillsonZWD15} 
deals with a problem slightly different from the minimum weight $m$-dominating set problem.
In the present paper, we give the first nontrivial approximation
algorithm for this problem (in Section~\ref{subs.dominating}),
and hence
the algorithm of Shi, Zhang, and Du~\cite{ShiZD15}
achieves the claimed approximation guarantee only when combined with our result.
In their journal version~\cite{ShiZMD17},
this fact is described.
In the second step of converting the $m$-dominating set into a
$k$-connected $m$-dominating set,
the approach of  \cite{ShiZD15,ShiZMD17}
and ours are different.
We give an $O(k^2)$-approximation algorithm for
 the augmentation
 problem,
 which is a special case of the optimization problem arising in the
 second step.
 From this algorithm, we derive an $O(k^2 \log k)$-approximation for the general
 case.
 On the other hand,  \cite{ShiZD15,ShiZMD17} gave a reduction from the general case
 to the edge-weighted subset $k$-connected subgraph problem.
 Their approach is weaker than our result
 in the fact that their reduction does not give a better algorithm for the
 augmentation problem.
 To be fair, their reduction approach has a merit that
  enables to use any algorithm for the edge-weighted subset $k$-connected subgraph problem
 while our analysis applies to a specific primal-dual algorithm.

Although this is not our main focus, 
we also present an improved analysis of the algorithm for the unweighted
$(2,m)$-CDS problem 
in Shang~et~al.~\cite{Shang:2007jg};
see Appendix~\ref{app.rect}.
Shi~et~al.~\cite{ShiZZW16} pointed out that the analysis of \cite{Shang:2007jg}
contains a flaw and the approximation factor given therein is not correct.
Shi~et~al.\ also gave a modified analysis, but our new analysis given in
Appendix~\ref{app.rect} attains a 
better approximation factor.

 \subsection{Organization}

 The remainder of this paper is organized as follows. 
 Section~\ref{sec.relatedworks} surveys related works.
 Section~\ref{sec.prelim} introduces preliminary facts used in this
 paper.
 Section~\ref{sec.simple} presents our $O(5^k k)$-approximation
 algorithm for the unweighted $(k,m)$-CDS problem, and
 Section~\ref{sec.primal-dual} provides our $O(k^2 \log
 k)$-approximation algorithm for the weighted $(k,m)$-CDS
 problem. Section~\ref{sec.conclusion} concludes the paper.
 Appendix~\ref{app.rect} rectifies the analysis on the algorithm of Shang~et~al.~\cite{Shang:2007jg}
 for the unweighted $(2,m)$-CDS problem.
 
\section{Related works}
\label{sec.relatedworks}

The study of the $(1,1)$-CDS problem for general graphs was initiated 
by Guha and Khuller~\cite{GuhaK98}. They presented an
$O(H(\Delta))$-approximation algorithm for the unweighted $(1,1)$-CDS
problem in graphs
with maximum degree $\Delta$,
where $H(\Delta)$ denotes the $\Delta$-th harmonic number.
They also
gave a reduction from the set cover to the unweighted $(1,1)$-CDS problem,
which shows that no polynomial-time algorithm achieves
an approximation factor $(1-\epsilon)H(\Delta)$
for any fixed $\epsilon \in (0, 1)$ unless
${\rm NP} \subseteq {\rm DTIME}[n^{O(\log \log n)}]$.
For the weighted $(1,1)$-CDS problem,
they gave an $2.613\ln n$-approximation
algorithm, where $n$ is the number of nodes in
the given graph.
This algorithm was improved by the same authors to $(1.35+\epsilon)\ln n$ in~\cite{GuhaK99}.

The unweighted $(1,1)$-CDS problem is NP-hard, even in unit disk graphs~\cite{ClarkCJ90}.
 Marathe~et~al.~\cite{MaratheBHRR95}
showed that the unweighted $(1,1)$-CDS problem in unit disk graphs admits a 10-approximation algorithm.
This has been improved by subsequent studies, and the current
best result is due to 
Cheng~et~al.~\cite{ChengHLWD03}, who gave a
polynomial-time approximation scheme
(PTAS) for the unweighted
$(1,1)$-CDS problem in unit disk graphs; note that the existence of a PTAS means
that for
any fixed constant $\epsilon >0$, there exists a 
($1+\epsilon$)-approximation algorithm that runs in polynomial time.

The first constant-approximation algorithm for
the weighted $(1,1)$-CDS problem in unit disk graphs was due to
Amb{\"{u}}hl~et~al.~\cite{AmbuhlEMN06}.
The current best approximation factor
for the same problem
is
$3+2.5\rho+\epsilon$, where $\rho$ is the approximation factor for the
edge-weighted Steiner tree in general graphs.
 This is achieved by
combining the $(3+\epsilon)$-approximation algorithm
due to Willson~et~al.\ \cite{WillsonZWD15}
for the minimum
weight 1-cover problem, and the $2.5\rho$-approximation
algorithm due to Zou~et~al.\ \cite{ZouLGW09}
for the node-weighted Steiner tree problem.
Currently it is known that $\rho \leq 1.39$~\cite{ByrkaGRS13}.
Zou~et~al.\ \cite{ZouWXLDWW11} claimed a $(5+\epsilon)$-approximation
algorithm for the weighted $(1,1)$-CDS problem, but their algorithm does
not seem to run in polynomial time for general node weights.

The first study of the $(k,m)$-CDS problem
was conducted by Dai and Wu~\cite{Dai:2006}.
They found several heuristic algorithms for the unweighted problem with
$k=m$ in unit disk graphs.
Thus far, several constant-approximation algorithms have been given for the
unweighted $(k,m)$-CDS problem in unit disk graphs~\cite{Shang:2007jg,WangTD09,WangKAGLZW13,Wang:2015}
For $k \leq 2$, Shang~et~al.\ \cite{Shang:2007jg} gave an 
approximation algorithm,
and it achieves a current best approximation factor for $k=1$.
As mentioned in Section~\ref{sec.ourresults},
Shi~et~al.~\cite{ShiZZW16}
pointed out a flow in the analysis of
Shang~et~al.\ \cite{Shang:2007jg} for $k=2$,
and gave its rectification.
Simultaneously Shi~et~al.\ presented an algorithm for the unweighted $(2,m)$-CDS problem on general
graphs. This algorithm achieves the current best approximation factor for
the unweighted $(2,m)$-CDS problem even on
unit disk graphs.
In Appendix~\ref{app.rect},
we present an improved analysis of Shang~et~al.~\cite{Shang:2007jg} for
$k=2$, but the obtained approximation factors are not better than 
Shi~et~al.
For $k=3$,
the current best approximation algorithm is 
due to
Zhang~et~al.\ \cite{Zhang16} (their algorithm is defined for general graphs).
Several previous papers
claimed constant-approximation algorithms for $k\geq 4$, but
Kim~et~al.~\cite{KimWLZW10} showed that
all of them had technical errors.
As far as we know, there was no known constant-approximation algorithm for the
unweighted $(k,m)$-CDS problem with $k\geq 4$ or for
the weighted $(k,m)$-CDS problem with $(k,m)\neq (1,1)$ before our study.
For the weighted problem, it was not even known whether a
constant-approximation algorithm exists for the problem of finding a
minimum weight $m$-dominating set in a unit disk graph when $m > 1$.

We note that 
distributed algorithms for
the unweighted $(1,1)$-CDS problem are also actively being studied.
Since this paper will focus on centralized algorithms,
we will only refer to a few of these
studies~\cite{Alzoubi02,FunkeKMS06,YuJYLC2015}.

 \section{Preliminaries}
 \label{sec.prelim}

 Let $G=(V,E)$ be an undirected graph with a node set $V$ and an edge
 set $E$.
 We denote $|V|$ by $n$.
 For $X \subseteq V$, let $G[X]$ denote the subgraph of $G$
 induced by $X$; i.e., its node set is $X$ and the edge set consists of
 the edges that join nodes in $X$.
 Throughout the paper, on the power set of $V$, we define maximality and minimality
with respect to inclusion.
 In other words, $X$ is minimal in a family $\Vfam \subseteq 2^V$
 if there is no $Y \in \Vfam$ with $Y \subset X$,
 and $X$ is maximal in $\Vfam$ if there is no $Y\in \Vfam$
 with $X \subset Y$.

 In a \emph{unit disk graph}, each node is placed on the two-dimensional
 Euclidean space,
 and two nodes are joined by an edge if and only if the distance between
 them is not larger than a unit length.
 The following property of unit disk graphs is well known,
 and it is used in \cite{Shang:2007jg}.

 \begin{lemma}\label{lem.kissingnum}
  Let $G=(V,E)$ be a unit disk graph. Let $v \in V$,
  and let $u_1,\ldots,u_6$ be distinct neighbors of $v$.
  Then $E$ includes an edge that joins two nodes in $\{u_1,\ldots,u_6\}$.
 \end{lemma}

 For the most part, our algorithms will require only the property stated
 in Lemma~\ref{lem.kissingnum}; the exception is the subroutine
 of computing an $m$-dominating set
 in our $O(k^2 \log k)$-approximation algorithm
 (presented in Section~\ref{subs.dominating}).
 
For $X \subseteq V$, we denote the set of neighbors of $X$ in $G$ by $\Gamma(X)$. In
other words, $\Gamma(X)=\{v \in V\setminus X \colon uv \in E \text{ for
some } u \in X\}$.
We also let $X^+$ denote $X \cup \Gamma(X)$ for any $X \subseteq V$.
For $X,T \subseteq V$, we simply denote $\Gamma(X) \cap T$ by
$\Gamma_T(X)$.
The following property of the function $\Gamma$
has been frequently used in previous works on the survivable networks
design (see e.g., \cite{Nutov12}).

\begin{lemma}\label{lem.submodularity}
 For any $X,Y,T \subseteq V$,
  the following hold:
 \begin{equation}\label{eq.submodular}
 |\Gamma_T(X)| + |\Gamma_T(Y)| \geq
  |\Gamma_T(X \cap Y)| + |\Gamma_T(X \cup Y)|,
 \end{equation}
 \begin{equation}\label{eq.negamodular}
 |\Gamma_T(X)| + |\Gamma_T(Y)| \geq
  |\Gamma_T(X \setminus Y^+)| + |\Gamma_T(Y \setminus X^+)|.
 \end{equation}
\end{lemma}

Let $X \subseteq V$.
For $T\subseteq V$,
$X$ is called a
 \emph{$T$-cut} if $X \subseteq T$ and  $X \neq \emptyset \neq T
 \setminus X^+$, and it is called a
\emph{Steiner $T$-cut} if
$X \cap T \neq \emptyset \neq T\setminus X^+$.
 Notice that a Steiner-$T$ cut is not necessarily a subset of $T$
 while a $T$-cut is so.
For $T \subseteq V$ and $r \in T$,
 a $T$-cut $X$
 (resp., a Steiner $T$-cut $X$)
 is called a $(T,r)$-cut
 (resp., Steiner $(T,r)$-cut)
 if $r\not\in X^+$.
A graph $G$ is $k$-connected
when it is connected even when any $k-1$ or fewer nodes are removed from the graph.
We note that a graph on at most $k$ nodes is $k$-connected by definition if it is a complete graph.
By Menger's theorem, a graph $G$ is $k$-connected if and only if $|\Gamma(X)|
\geq k$ for any nonempty $X \subseteq V$ with $X^+ \neq V$.
A subset $T$ of $V$ is $k$-connected if and only if $|\Gamma_T(X)|\geq k $
for any $T$-cut $X$
(recall that the $k$-connectivity of $T$ is defined by the
$k$-connectivity of $G[T]$).

Our algorithms take the same approach
proposed in
previous studies;
they compute an $m$-dominating set in the first step, and
increase its connectivity to $k$.
In our algorithms, this is achieved by solving
the following problems.

 \begin{definition}[Augmentation problem]
  \label{def.augmentation}
 Given an undirected graph $G=(V,E)$, a nonnegative weight $w(v)$ for
 each node $v \in V$, and a $(k-1,m)$-CDS
 $T$ of $G$, find $S \subseteq V \setminus T$
 that minimizes $\sum_{v \in S}w(v)$, 
 subject to the condition that $T \cup S$ is $k$-connected in $G$ (i.e.,
  $G[T\cup S]$ is a $k$-connected graph).
 \end{definition}

 For $T \subseteq V$, a path $P$ is called a \emph{$T$-path}
 if both end nodes of $P$ are included in $T$, and no inner nodes of $P$
 are included in $T$. 

A family $\Lfam$ of subsets of $V$ is said to be 
\emph{laminar}
if any $X,Y \in \Lfam$ satisfy $X \subseteq Y$, $Y \subseteq X$,
or $X \cap Y = \emptyset$.
$\Lfam$ is said to be
\emph{strongly laminar}
if any $X,Y \in \Lfam$ satisfies $X \subseteq Y$, $Y \subseteq X$,
or $X \cap Y^+ = \emptyset = X^+ \cap Y$.
Let $\Lfam$ be a laminar family.
If $X$ is minimal in $\Lfam$,
we call $X$ a \emph{leaf} of $\Lfam$.
For some $Y \in \Lfam$,
if $X$ is a maximal member of $\Lfam$ 
subject to $X \subset Y$,
then we say that $X$ is a \emph{child} of $Y$.

\section{Simple algorithm for the unweighted problem}
\label{sec.simple}

In this section, we present an $O(k 5^k )$-approximation algorithm for the unweighted
$(k,m)$-CDS problem
with $m \geq k$.
We may assume $G$ is $k$-connected,
since otherwise $G$ has no $(k,m)$-CDS.

Our algorithm computes an $m$-dominating set by 
applying the constant-approximation algorithm given in
\cite{Shang:2007jg},
and then increases its connectivity by adding nodes.
Let us explain this in more details.
Let $T'$ be an $m$-dominating set.
Note that $|T'|\geq m \geq k$.
First,
we choose $k$ roots $r_1,\ldots,r_k \in T'$.
For each $i=1,\ldots,k$,
we compute a node set $S_i \subseteq V \setminus T'$
such that the connectivity between $r_i$ and each node $v \in T'\setminus
\{r_i\}$ is at least $k$ in the graph $G[T' \cup S_i]$
(i.e., $G[T' \cup S_i]$ includes $k$ internally disjoint paths between $r_i$
and each node $v \in T' \setminus \{r_i\}$).
Our solution is $T' \cup (\bigcup_{i=1}^k S_i)$.

Let us see that $T' \cup (\bigcup_{i=1}^k S_i)$ is a $(k,m)$-CDS.
We can see that each pair of nodes in $T'$ is $k$-connected in
$G[T' \cup (\bigcup_{i=1}^k S_i)]$
(we say in this case that \emph{$T'$ is $k$-connected in $G[T' \cup (\bigcup_{i=1}^k S_i)]$}),
as follows.
If $u,v \in T'$ is not $k$-connected in $G[T' \cup (\bigcup_{i=1}^k
S_i)]$,
there exists $U \subseteq T' \setminus \{u,v\}$
with $|U|\leq k-1$
such that $u$ and $v$ are not connected after removing $U$.
Since $|U| \leq k-1$, there exists $r_i$ such that $r_i \notin U$.
This $r_i$ is not connected to $u$ or $v$ in the graph after removing
$U$, which contradicts the definition of $S_i$.
This means that $T'$ is $k$-connected in
$G[T' \cup (\bigcup_{i=1}^k S_i)]$.
Moreover, from the facts that $T'$ is $m$-dominating and $m \geq k$,
it follows that $T'\cup (\bigcup_{i=1}^k S_i)$ is also $k$-connected.
Indeed, we can observe that any node in $\bigcup_{i=1}^k S_i$ is not disconnected from other
nodes even after $k-1$ nodes are removed from $G[T'\cup (\bigcup_{i=1}^k S_i)]$.
Therefore,  $T' \cup (\bigcup_{i=1}^k S_i)$ is a $(k,m)$-CDS.

We show that $|S_i| \leq (5^k-1)|T'|$ holds for each $i=1,\ldots,k$.
Since $|T'| =O(1) \cdot \opt$,
this implies that $|T' \cup (\bigcup_{i=1}^k S_i)| \leq O(k5^k) \cdot \opt$.
Hence $T' \cup (\bigcup_{i=1}^k S_i)$  achieves the approximation factor $O(k5^k)$.

To compute $S_i$, we repeat solving the augmentation problem.
Here, the definition of the augmentation problem is slightly different
from Definition~\ref{def.augmentation};
$T$ is not necessarily $(k-1)$-connected here.
Instead, the connectivity between $r_i$ and each node $v \in T\setminus
\{r_i\}$ is $j-1$ in $G[T]$ for some integer $j$
with $1 \leq j \leq k$,
and we are required to compute a minimum size $S_{i,j}$
such that
 connectivity between $r_i$ and each node in $T\setminus \{r_i\}$ is $j$ in $G[T
 \cup S_{i,j}]$.
We call this version of the augmentation problem by the \emph{rooted
augmentation problem}.
For the rooted augmentation problem,
 we prove that it is possible to find a solution $S_{i,j}$
 such that $|S_{i,j}|\leq 4 |T|$.
 When we compute $S_{i,j}$, we solve the problem
 with setting $T$ to $T' \cup (\bigcup_{j'=1}^{j-1}S_{i,j'})$.
 This gives $|T' \cup (\bigcup_{j'=1}^{j} S_{i,j'})|
 \leq |S_{i,j}| + |T' \cup (\bigcup_{j'=1}^{j-1} S_{i,j'})| \leq 5|T'
 \cup (\bigcup_{j'=1}^{j-1} S_{i,j'})|$
 holds for any $j=1,\ldots,k$.
 $S_i$ is defined as $\bigcup_{j'=1}^{k} S_{i,j'}$.
 Then $|T' \cup S_i|= |T'\cup (\bigcup_{j=1}^k S_{i,j})| \leq 5^k |T'|$,
 indicating that $|S_i| \leq (5^k -1)|T'|$.
 
 In the rest of this section, we discuss the rooted augmentation
 problem.
 For simplicity,
 we denote $r_i, j, S_{i,j}$ by $r$, $k$, and $S$, respectively.
 Namely, the connectivity between $r$ and each $v \in T \setminus \{r\}$ is $k-1$
 in $G[T]$, and we are asked to find a minimum size $S \subseteq V
 \setminus T$
 such that the connectivity between $r$ and each $v \in T \setminus \{r\}$ is at least $k$ in
 $G[T \cup S]$.

 By the definition of $T$,
all $(T,r)$-cuts $X$ satisfy $|\Gamma_T(X)|\geq k-1$.
We say that a $(T,r)$-cut $X$ is a \emph{demand cut}
if $|\Gamma_T(X)|=k-1$.
We also say that a $T$-path $P$ \emph{covers} a demand cut $X$
if one end node of $P$ is in $X$, and the other end node is in
$T\setminus X^+$.

   \begin{lemma}\label{lem.Ucut}
    Let $m \geq k$,
    and let $T$ be an $m$-dominating set of a graph $G=(V,E)$
    such that the connectivity between $r \in T$ and each node in $T\setminus \{r\}$ is
    at least $k-1$ in $G[T]$.
  Let $S \subseteq V \setminus T$.
  If a $(T\cup S,r)$-cut $X$ in $G$ satisfies $|\Gamma_{T \cup
  S}(X)| \leq k-1$,
  then $X \cap T$ is a demand cut covered by no $T$-path in $G[T \cup S]$.
 \end{lemma}
   \begin{proof}
    Since $T$ is an $m$-dominating set, if $X \subseteq S$, 
   then $|\Gamma_{T\cup S}(X)| \geq |\Gamma_T(X)|\geq m \geq k$ holds,
   which contradicts the assumption.
   Hence $X \cap T$ is a $(T,r)$-cut.
  Since $\Gamma_T(X\cap T) \subseteq \Gamma_T(X) \subseteq  \Gamma_{T\cup S}(X)$,
    we have $|\Gamma_T(X\cap T)|= | \Gamma_{T\cup S}(X)|=k-1$.
    This means that $X\cap T$ is a demand cut.
    In particular, since $\Gamma_T(X\cap T)=\Gamma_{T \cup S}(X)$,
    no $T$-path in $G[T\cup S]$ covers $X\cap T$.
   \end{proof}

   The following lemma presents a characterization of feasible solutions
   for the rooted augmentation problem.
    \begin{lemma}
     \label{lem.feasibility}
         Let $m \geq k$,
    and let $T$ be an $m$-dominating set of a graph $G=(V,E)$
    such that the connectivity between  $r \in T$ and each node in $T\setminus \{r\}$ is at least
     $k-1$ in $G[T]$.
    $S \subseteq V\setminus T$ is feasible for the rooted augmentation
    problem with $G$, $T$, and $r$
    if and only if every demand cut $X$ is covered by a
    $T$-path in $G[S \cup T]$.
    \end{lemma}
     \begin{proof}
      If there exists a demand cut $X$ covered by no $T$-path in $G[S\cup T]$,
      then, by Lemma~\ref{lem.Ucut},
      any $(T\cup S,r)$-cut $X$ in $G$ satisfies $|\Gamma_{T \cup S}(X)| \geq k$.
      The ``if'' part follows from this fact.

      Let us prove the ``only if'' part.
     Let $X$ be a demand cut,
      and $v \in X$.
     If $X$ is covered by 
     no $T$-path in $G[S \cup T]$,
     then each path connecting $v$ and $r$  in $G[S \cup T]$ 
      includes some node in $\Gamma_T(X)$. Since $|\Gamma_T(X)|\leq k-1$,
     this implies that the connectivity between $v$ and $r$ is at most
     $k-1$ in $G[S\cup T]$.  Therefore, $S$ is not feasible.
     \end{proof}

  For $S \subseteq V \setminus T$,
  we let $\Dfam(r,S)$ denote
  the family of all demand cuts
   covered by no $T$-path in $G[T \cup S]$.
  Lemma~\ref{lem.feasibility} implies that $S$ is feasible if and only if $\Dfam(r,S)=\emptyset$.

The following lemma was used in previous studies.

\begin{lemma}[\cite{Shang:2007jg,WangKAGLZW13,Wang:2015}]\label{lem.path}
 Let $m \geq k$, and 
 let $T$ be an $m$-dominating set  of a $k$-connected graph $G=(V,E)$
 such that the connectivity between $r$ and each node in $T \setminus \{r\}$
 is at least $k-1$ in $G[T]$.
 For every demand cut $X$ of $G$ with respect to $T$
 (i.e., $X$ is a $(T,r)$-cut and $|\Gamma_T(X)|=k-1$ in $G$), $G$ includes a $T$-path that covers $X$
 and contains at most two inner nodes.
\end{lemma}

From these observations, we can consider
the following simple algorithm for the rooted augmentation
problem.
Initialize a solution $S$ to an empty set.
If $S$ is not feasible for the rooted augmentation problem,
there exists a demand cut $X \in \Dfam(r,S)$.
The algorithm 
chooses such a demand cut $X$, and adds to $S$
the inner nodes of a
$T$-path covering $X$ that is guaranteed by
Lemma~\ref{lem.path}.
The algorithm repeats this until $T\cup S$ becomes feasible.
In fact, this is almost same as the algorithms proposed in
\cite{Shang:2007jg,Wang:2015}
for $k\leq 3$ while \cite{Shang:2007jg,Wang:2015} consider
the augmentation problem instead of the rooted augmentation problem.
Every iteration of this algorithm adds at most two nodes to $S$.
Hence,
in order to obtain an approximation guarantee for this algorithm,
it is critical to analyze how many iterations are
required to ensure that $T \cup S$ is feasible.

We analyze the number of iterations for a general connectivity requirement $k$.
To do this, we make a slight modification to the algorithm.
We restrict the demand cut
$X$ that is chosen in each iteration, as follows.
Instead of an arbitrary demand cut in $\Dfam(r,S)$,
our algorithm
always chooses a minimal of such cuts.
This procedure is described in detail as Algorithm~\ref{alg.simple}.

 \begin{algorithm}
  \caption{Algorithm for the rooted augmentation problem}
  \label{alg.simple}
  \begin{algorithmic}
   \REQUIRE
   integers  $m,k$ with $m \geq k$,
   an undirected graph $G=(V,E)$,
   an $m$-dominating set $T \subseteq V$,
   and a root $r \in T$ such that
   $r$ and each node in $T \setminus \{r\}$ is
   $(k-1)$-connected in $G[T]$
   \ENSURE $S \subseteq V \setminus T$ such that
   $r$ and each node in $T \setminus \{r\}$ is $k$-connected in $G[T\cup S]$
   \STATE $S \longleftarrow \emptyset$
   \WHILE{$S$ is not feasible}
   \STATE $X \longleftarrow$ a minimal cut in $\Dfam(r,S)$
   \STATE $P \longleftarrow$ a minimum-length $T$-path that covers $X$
   \STATE add the inner nodes in $P$ to $S$
   \ENDWHILE
   \STATE output $S$
  \end{algorithmic}
 \end{algorithm}

 Let us explain that each iteration of Algorithm~\ref{alg.simple} can be done in polynomial time.
 A minimum-length $T$-path $P$ that covers $X$
 can be computed by a shortest path algorithm.
 For computing a minimal cut in $\Dfam(r,S)$,
 we first compute a minimal $(T\cup S, r)$-cut $X'$ that minimizes
 $|\Gamma_{T\cup S}(X')|$ in $G[T \cup S]$. This can be found by applying 
 a max-flow algorithm repeatedly to the graph $G[T\cup S]$ with unit
 node-capacities except on the sink and the source
 while setting the source to each node in $(T
 \cup S) \setminus \{r\}$ and the sink to $r$.
 Since $S$ is not feasible, $|\Gamma_{T \cup S}(X')|=k-1$ holds.
 Lemma~\ref{lem.Ucut} indicates $X' \cap T \in \Dfam(r,S)$.
 Indeed, $X' \cap T$ is a minimal cut in $\Dfam(r,S)$. To see this, 
 suppose that there exists $X'' \in \Dfam(r,S)$ with $X'' \subset X'
 \cap T$.
 Let $S'$ be the subset of $S$ that consists of nodes on
 $T$-paths  in the graph $G[S\cup T]$ having end nodes in $X''$. Then, since $X''$ is 
 covered by no $T$-path, $|\Gamma_{T\cup S}(X'' \cup S')|=k-1$, which
 contradicts the minimality of $X'$.

In the following theorem, we show that $O(|T|)$ iterations are sufficient to ensure that
Algorithm~\ref{alg.simple} computes a feasible solution.

 \begin{theorem}
  \label{thm.simple-alg}
 Algorithm~\ref{alg.simple} 
  outputs a solution after $2|T|-3$ iterations.
 \end{theorem}

Before presenting a proof of Theorem~\ref{thm.simple-alg},
let us give a consequence of Theorem~\ref{thm.simple-alg}.
Each iteration of
 Algorithm~\ref{alg.simple} adds at most two nodes to the solution.
 Hence, Theorem~\ref{thm.simple-alg} immediately implies
 that
 Algorithm~\ref{alg.simple} outputs a solution $S$ such that $|S|\leq
 2(2|T|-3)\leq 4|T|$.
 By the discussion above, this presents the following corollary.

   \begin{corollary}\label{cor.unweighted}
    There exists an $O(k5^k)$-approximation algorithm for 
    the unweighted $(k,m)$-CDS problem in unit disk graphs
    if $m \geq k$.
   \end{corollary}

   We note that our algorithm for the rooted augmentation problem
   does not rely on any property specific to unit disk graphs, and so
  the result in Corollary~\ref{cor.unweighted}
   can be extended to any graph class that admits a 
   constant-approximation algorithm for
   finding a minimum $m$-dominating set.
   One such example
   is the class of 
   bounded clique- and tree-width graphs~\cite{CicaleseCGMV14}.

 Now, we prove Theorem~\ref{thm.simple-alg}.
 \begin{lemma}
  \label{lem.crossing}
  Let $r \in T$ and $S \subseteq V \setminus T$.
  Let $X,Y \in \Dfam(r,S)$. If $X$ is minimal in $\Dfam(r,S)$, then
  $X\cap Y=\emptyset$ or $X \subseteq Y$ holds.
  \end{lemma}
\begin{proof}
 Suppose that
 some pair of $X$ and $Y$ violates this claim.
  Then, $X\cap Y \neq \emptyset$ holds.
  The minimality of $X$ implies that $Y \subset X$ does not
 hold.
  Hence, 
 $\emptyset \neq X \cap Y \subset X \subset X \cup Y$ holds.
 Also, $r \not\in (X \cup Y)^+$ follows from
 $r \not \in X^+$ and $r \not \in Y^+$, and
 $r \not\in (X \cap Y)^+$ follows, because
 $(X \cap Y)^+ \subseteq (X\cup Y)^+$.
 Therefore, both $X \cap Y$ and $X\cup Y$ are $(T,r)$-cuts.

 For each $T$-path $P$ in $G[T\cup S]$,
 we add an edge joining two end nodes of $P$ to $G[T]$.
 Let $G'$ denote the graph with the node set $T$ obtained by this operation. 
 Consider inequality \eqref{eq.submodular}, where $\Gamma_T$ is
 defined with respect to the graph $G'$.
 The left-hand side of  \eqref{eq.submodular} is exactly $2(k-1)$,
 because $X, Y \in \Dfam(r,S)$,
 and the right-hand side is at least $2(k-1)$,
 because $X\cap Y$ and $X \cup Y$ are $(T,r)$-cuts.
 Hence, the inequality holds with equality, and
 $|\Gamma_T(X \cap Y)| =|\Gamma_T(X \cup Y)| =k-1$.
 This means that both $X \cap Y$ and $X \cup Y$ belong to $\Dfam(S,r)$.
 This contradicts the minimality  of $X$.
\end{proof}

Define $\Afam$ as the family of demand cuts
chosen in the while loop of Algorithm~\ref{alg.simple}.
Theorem~\ref{thm.simple-alg} is implied by the following lemma.

 \begin{lemma}
  \label{lem.laminarity}
 $|\Afam| \leq 2|T|-3$.
 \end{lemma}
 \begin{proof}
  We prove that $\Afam$ is a laminar family
  on $T \setminus \{r\}$.
  The lemma follows from this, because the size of a laminar family
  on the set of cardinality $|T|-1$ is at most $2|T|-3$.
  
Suppose that there exist $X,Y \in \Afam$ such that
 $X\cap Y \neq \emptyset$, $X \setminus Y \neq \emptyset$, and $Y
  \setminus X \neq \emptyset$.
  We may assume without loss of generality that
  $X$ is chosen in an earlier iteration than that in which $Y$ is chosen.
  Let $S$ denote the subset at the beginning of the iteration during which $X$ is chosen.
  Then
  both $X$ and $Y$ belong to $\Dfam(r,S)$,
  and $X$ is minimal in $\Dfam(r,S)$.
  However, Lemma~\ref{lem.crossing} shows that
  $X \cap Y = \emptyset$.
    This contradicts the definitions of $X$ and $Y$.
 \end{proof}

\section{Primal-dual algorithm for the weighted problem}
\label{sec.primal-dual}

In this section, we present an $O(k^2 \log k)$-approximation algorithm
for the weighted $(k,m)$-CDS problem.
Our algorithm for this problem is also based on a subroutine that
solves the augmentation problem. We show that there is a
constant-approximation algorithm for the augmentation problem with general
node weights.
This algorithm is based on the primal-dual method, which is a technique
for computing an approximate solution from a linear programming (LP) relaxation of the problem.
Before introducing the primal-dual algorithm, we consider the
weighted $m$-dominating set problem, which 
demands a minimum weight $m$-dominating set;
we prove that the problem admits a constant-approximation algorithm.

\subsection{Approximation algorithm for the weighted  $m$-dominating set problem}
\label{subs.dominating}

Our algorithm reduces the weighted $m$-dominating set problem to the following
geometric problem.

 \begin{definition}[Disk multicover problem]
 We are given a set $P$ of points and a set $D$ of disks on the
 Euclidean plane,
  a demand $d(p)$ for each point $p \in P$,
 and a nonnegative weight $w(i)$ for each disk $i \in D$.
  A subset $D'$ of $D$
is called a disk multicover if
 each point $p \in P$ is contained in at least $d(p)$ disks in $D'$.
  The problem requires finding a disk multicover $D'$ that minimizes the weight $\sum_{i \in D'}w(i)$.
  \end{definition}

When $d(p)=1$ for all $p \in P$, this is called the
\emph{disk cover problem}.
We write $p \in i$ if a point $p$ is included in a disk $i$.

Bansal and Pruhs~\cite{Bansal:2012ik} presented a constant-approximation algorithm for the disk
multicover problem.
Their algorithm is an LP-rounding algorithm. That is to say, their
algorithm first solves the following LP relaxation of the problem:
\begin{equation}
 \label{eq.lp-diskcover}
 \begin{array}{ll}
  \text{minimize} & \sum_{i \in D} w(i) x(i)\\
  \text{subject to} & \sum_{i \in D: p \in i} x(i) \geq d(p), \ \forall  p \in P,\\
  &  0 \leq x(i) \leq 1, \  \forall i \in D,
 \end{array}
\end{equation}
then it computes a disk multicover $D'$ 
that satisfies $\sum_{i \in D'}w(i) = O(1) \cdot \sum_{i \in D}
w(i)x(i)$ for an optimal solution $x$ to \eqref{eq.lp-diskcover}.

When $m=1$, the problem of finding a minimum weight $m$-dominating set
in a unit disk graph
can be reduced to the disk cover problem, as follows.
Define $D$ as the set of unit disks corresponding to the nodes in the
unit disk graph, and define $P$ as the set of the centers of the disks.
The weight $w(i)$ of a disk $i \in D$ is defined as the weight of the
corresponding node in the graph.
For each point $p$,
a disk cover in this instance includes at least one disk
that contains $p$. 
This means that the set of nodes corresponding to the disks in the 
disk cover
is a $1$-dominating set of the graph.

From the weighted $m$-dominating set problem with $m \geq 2$,
we can similarly define an instance of the disk multicover problem;
$D$, $P$, and $w$ are defined in the same way,
and the demand $d(i)$ of each disk $i \in D$ is defined as $m$.
By solving this instance, we can obtain
an $m$-dominating set in the unit disk graph.
However, the minimum weight of disk multicovers in the obtained instance
is possibly too large, compared with the minimum weight of the $m$-dominating sets.
To see this, let $D'$ be a disk multicover in the obtained instance of the
disk multicover problem.
The constraint in the disk multicover
problem demands that each point $i \in D'$ is included in at least
$d(i)$
disks in $D'$. On the other hand, in the weighted $m$-dominating set problem,
if a solution includes
a node $i$,
it is feasible even if
it does not contain $d(i)$ neighbors of $i$.
In other words, the constraint of the disk multicover problem
in the constructed instance is stronger than the one demanded in the original
instance of the weighted $m$-dominating set problem.
Accordingly, there does not seem to exist a straightforward reduction
from the weighted $m$-dominating set problem to the disk multicover problem.

Nevertheless, we show that
the weighted $m$-dominating set problem in a unit disk graph
can be approximated via an algorithm for 
the disk multicover problem.
Our algorithm first solves the following LP relaxation of the
weighted $m$-dominating set problem:
\begin{equation}
 \label{eq.lp-dominatingset}
 \begin{array}{lll}
  \text{minimize} & \sum_{v \in V} w(v) x(v)\\
  \text{subject to} & \sum_{u \in \Gamma(v)\setminus S} x(u) \geq (m-|S|)
   (1-x(v)), & \forall v \in V, \forall S \subseteq \Gamma(v),\\
  &  0 \leq x(v) \leq 1, & \forall v \in V.
 \end{array}
\end{equation}

Although \eqref{eq.lp-dominatingset}  has an exponential number
of constraints, the separation can be done in polynomial time.
Namely, given $x$, we can judge whether $x$ is a feasible solution for
\eqref{eq.lp-dominatingset}, as follows.
For each $v$, sort the neighbors $u_1,\ldots,u_i \in \Gamma(v)$ so that
$x(u_1) \leq \cdots \leq x(u_i)$.
If $\sum_{j=1}^{i-m'} x(u_{j}) \geq (m-m')(1-x(v))$  for each
$m'=0,\ldots,m$,
then the constraint defined from $v$ and all $S \subseteq \Gamma(v)$ is
satisfied by $x$.
If $\sum_{j=1}^{i-m'} x(u_{j}) \geq (m-m')(1-x(v))$ does not hold for some
$m'=0,\ldots,m$,
then the constraint defined from $v$ and
$S:=\{u_{i-m'+1},\ldots,u_{i}\}$
($S:= \emptyset$ when $m'=0$)
is
violated by $x$.
Hence, the separation can be done by checking whether the $m+1$
inequalities hold for each node $v$.
Therefore, the ellipsoid method can be used to solve \eqref{eq.lp-dominatingset} in polynomial time.

Let $x^*$ denote an optimal solution for \eqref{eq.lp-dominatingset}.
We define $U$ as $\{v \in V \colon x^*(v)\leq 1/2 \}$.
Our algorithm invokes
an algorithm for the disk multicover problem after the input is set as follows.
$D$ is defined as
the set of disks corresponding to the nodes in $U$,
and $P$ is defined as the set of the centers of the disks in $D$.
The demand $d(p)$ of the point $p \in P$
corresponding to a node $u \in U$
is defined as $m - |\Gamma(u) \setminus U|$.
We solve the obtained instance of the disk multicover problem
by using an LP rounding algorithm based on \eqref{eq.lp-diskcover}, such as
the algorithm of Bansal and Pruhs~\cite{Bansal:2012ik}.
Let $D'$ be the set of nodes corresponding to the disks in the obtained approximate solution.
Our algorithm outputs $D' \cup (V \setminus U)$ as an approximate
solution for the weighted $m$-dominating set problem.

\begin{theorem}
 There exists a $2\alpha$-approximation algorithm for the weighted $m$-dominating set problem in a
 unit disk graph if an algorithm for the disk multicover problem
 computes a solution of weight at most
 $\alpha \cdot \sum_{i \in D} w(i)x(i)$ for an optimal solution $x$ to
 \eqref{eq.lp-diskcover}.
\end{theorem}
\begin{proof}
 We prove that our algorithm given above is a $2\alpha$-approximation algorithm.
 Let $\opt$ denote the minimum weight of the $m$-dominating sets,
 and let $x^*$ denote an optimal solution for \eqref{eq.lp-dominatingset}. Then,
 $\sum_{v \in V}w(v)x^*(v) \leq \opt$ holds
 because
 \eqref{eq.lp-dominatingset} relaxes 
  the weighted $m$-dominating set problem 
 (i.e., the incidence vector of each $m$-dominating set is a feasible solution to
 \eqref{eq.lp-dominatingset}).

For each $v\in U$, let $p_v$ and $i_v$ respectively denote the point in $P$ and the
 disk in $D$ corresponding to $v$.
 Define $\bar{x} \in [0,1]^{D}$ by $\bar{x}(i_v)=2x^*(v)$ for each $v \in
 U$.
 Then, for each $v \in U$,
 \begin{align*}
 \sum_{i \in D: p_v\in i}\bar{x}(i)
  &= 2 \sum_{u \in \Gamma_U(v) \cup \{v\}}x^*(u)\\
  &\geq 2 (m-|\Gamma(v)\setminus U|)(1-x^*(v))\\
  &\geq m-|\Gamma(v)\setminus U|\\
  &=d(p_v)
 \end{align*}
 holds, where the first inequality follows from the constraint
 of \eqref{eq.lp-dominatingset}, and the second inequality follows from
 $x^*(v) \leq 1/2$.
Hence, $\bar{x}$ is a feasible solution to \eqref{eq.lp-diskcover}.
 The algorithm for the disk multicover problem computes a solution $D'
 \subseteq D$ such that $\sum_{i \in D'}w(i) \leq \alpha \sum_{i \in D}w(i)\bar{x}(i)
 \leq 2\alpha\sum_{v \in U}w(v)x^*(v)$ for some constant $\alpha$.
 On the other hand, $\sum_{i \in V \setminus U}w(i) <
 2\sum_{v \in V\setminus U}w(v)x^*(v)$, because $x^*(v) > 1/2$ for
 each $v \in V\setminus U$.
 Therefore, $\sum_{v \in D' \cup (V\setminus U)}w(v)
 \leq 2 \alpha\sum_{v \in V}w(v)x^*(v) \leq 2\alpha \opt$.

 Let $u$ be a node that is not contained in $D' \cup (V\setminus U)$.
 Then, $u$ has $d(p_u)=m-|\Gamma(u)\setminus U|$ neighbors in $D'$
 and $|\Gamma(u)\setminus U|$ neighbors in $(V\setminus U)$.
 Therefore, $D' \cup (V\setminus U)$ is an $m$-dominating set.
 \end{proof}

 A drawback to our algorithm is that it requires the ellipsoid method,
 which tends to be slow in practice.
 When $m \in \{k,k+1\}$,
 this can be avoided by using the above-mentioned straightforward reduction to the disk multicover problem.
 Recall that the reduction does not work in general
 because
 a node in an $m$-dominating set $S$ may not have $m$ neighbors in $S$.
 However, when $S$ is $k$-connected for some $k \geq m-1$,
 each node $v \in S$ has $m-1$ neighbors in $S$. Note that $v$ is not
 counted as a neighbor of $v$.
 Hence, the minimum weight of the disk multicovers can be bounded by the
 minimum weight of the $(k,m)$-CDSs,
 and the straightforward reduction gives an $m$-dominating set
 that has a weight within a constant factor of the minimum weight of
 the $(k,m)$-CDSs.

 \subsection{Algorithm for the augmentation problem}
 \label{ssec.augmentation}
 
First, let us give an overview of our algorithm for the augmentation
problem.
When the connectivity requirement $k$ is equal to one, 
the augmentation problem is known as the \emph{node-weighted Steiner
tree problem}.
For general graphs, it is hard to approximate this problem within a factor of $o(\log n)$,
because it extends the set cover problem~\cite{KleinR95}.
However, in unit disk graphs, there is a constant-approximation
algorithm for the node-weighted Steiner tree problem.
Zou~et~al.~\cite{ZouLGW09} proved the existence of such an algorithm from the 
fact that
any unit disk graph has a spanning tree of maximum degree five.
This property of unit disk graphs is well known, and it can be shown by
using the following observation: if there is a node $v$ of degree more than
five in a spanning tree, then by Lemma~\ref{lem.kissingnum},
there is an edge $uu'$ that joins two neighbors $u$ and $u'$ of $v$.
Replacing the edge $vu$ by another edge $uu'$ transforms the spanning tree into
another spanning tree in which the degree of $v$ is decreased by one
(to ensure the existence of a spanning tree of maximum degree five, we must consider the degree of $u'$, because it is increased by the operation).
This approach cannot be directly extended to the general connectivity
requirement $k$, because this operation does not preserve the connectivity
of a graph.
To see this, consider the graph on seven nodes $u,
v_1,\ldots,v_6$ such that
$v_1\ldots,v_6$ form a cycle of length six, and $u$ is adjacent to each of
$v_1,\ldots,v_6$.
This graph is 3-connected, and the degree of $u$ is six. To decrease the
degree of $u$, replace one edge $uv_i$ by another edge $v_{i-1}v_i$, and 
then $v_i$ will have only two neighbors; hence, the
connectivity of the graph has been decreased to two.

Nevertheless, we will show that Lemma~\ref{lem.kissingnum} can be used to show 
that the augmentation problem admits a better approximation algorithm for unit disk graphs
than it does for general graphs. We will use the lemma in the framework
of the primal-dual method,
which has been applied successfully to many
network design problems~\cite{goemans1997primal}.
Our algorithm repeats growing several dual
variables simultaneously in an LP relaxation.
This approach has been used in the augmentation
problem with node weights~\cite{ChekuriEV12,Fukunaga15},
but 
its approximation factor depends on the number of nodes.
This is because the approximation factor is decided by the number of
dual variables that
are grown simultaneously in a single constraint.
In our LP relaxation of the augmentation problem, each dual variable
corresponds to
a demand cut, and each constraint corresponds to
a node in the given graph.
Since this number cannot be bounded, the primal-dual
method does not achieve a good approximation factor for
general graphs, but we will show that this number can be bounded in unit disk graphs, due to Lemma~\ref{lem.kissingnum}.

Let us explain the detail of our algorithm.
Consider an instance of the augmentation problem with a graph $G=(V,E)$
and an $(m,k-1)$-CDS $T$ of $G$.
Recall that $S \subseteq V \setminus T$ is defined to be feasible for the augmentation
problem
if $G[T \cup S]$ is $k$-connected.
Indeed,
as did in Section~\ref{sec.simple},
from the assumptions that $T$ is $m$-dominating and $m\geq k$,
we can see that $S$ is feasible if and only if
$T$ is $k$-connected in
$G[T\cup S]$.
Since $T$ is $(k-1)$-connected,
each Steiner $T$-cut $X$ satisfies $|\Gamma_T(X)| \geq k-1$ in $G$.
By the Menger's theorem, $S\subseteq V \setminus T$ is feasible for the augmentation problem
if and only if $S \cap \Gamma(X) \neq \emptyset$ for 
each Steiner $T$-cut $X$ with $|\Gamma_T(X)| = k-1$.

We choose a root node $r \in T$, and consider
the problem of finding a minimum weight node set
$S \subseteq V \setminus T$ such that
$S \cap \Gamma(X)\neq \emptyset$ holds for 
every Steiner $(T,r)$-cut $X$ with $|\Gamma_T(X)|=k-1$.
By solving this problem for different $k$ roots and
outputting the union of the obtained solutions,
we can solve the augmentation problem.

 For the remainder of this subsection, we fix a root node $r \in T$.
We say that a Steiner $(T,r)$-cut $X$ is a \emph{demand cut}
if $|\Gamma_T(X)|=k-1$.
 (Note that this is slightly different from the definition in
 Section~\ref{sec.simple}).
We denote by $\Dfam$ the family of all demand cuts.
 Observe that $S$ is a feasible solution for the current problem defined
 from $r$
 if and only if $\Gamma(X)\cap S\neq \emptyset$ for each demand cut $X$.
An LP relaxation of the problem can be formulated as follows:
\begin{equation}
 \label{eq.primal-lp}
 \begin{array}{lll}
  \text{minimize} & \sum_{v \in V \setminus T} w(v) x(v)\\
  \text{subject to} & \sum_{v \in \Gamma(X)\setminus T} x(v) \geq 1, 
   & \forall X \in \Dfam,\\
  &  x(v) \geq 0,  & \forall v \in V\setminus T.
 \end{array}
\end{equation}
The dual of this LP is 
\begin{equation}
 \label{eq.dual-lp}
 \begin{array}{lll}
  \text{maximize} & \sum_{X \in \Dfam} y(X)\\
  \text{subject to} & \sum_{X \in \Dfam: v \in \Gamma(X)} y(X) \leq w(v),
   & \forall v\in V \setminus T,\\
  &  y(X) \geq 0, &  \forall X \in \Dfam.
 \end{array}
\end{equation}
We say that
a node $v \in V \setminus T$ 
\emph{covers} a demand cut $X$
if $v \in \Gamma(X)$,
and a node set $S$
covers $X$ if there exists a node $v \in S$ that covers $X$.

A subfamily $\Ffam$ of $\Dfam$ is called \emph{uncrossable} when 
any $X,Y \in \Ffam$
satisfy $X \cap Y, X\cup Y \in \Ffam$ or
$X \setminus Y^+, Y \setminus X^+\in
\Ffam$.
$\Ffam$
is called \emph{$T$-intersecting}
 if
any $X,Y \in \Ffam$
with $X\cap Y \cap T \neq \emptyset$ satisfy $X \cap Y,X\cup Y \in
\Ffam$.
We will show that the augmentation problem admits  a
constant-approximation algorithm when $\Dfam$ is $T$-intersecting uncrossable.

In general, the family $\Dfam$ is not $T$-intersecting uncrossable.
If $\Dfam$ is not $T$-intersecting uncrossable, 
we use a decomposition result given by Nutov~\cite{Nutov12}.
Namely, our algorithm finds a $T$-intersecting uncrossable subfamily
$\Ffam$
of $\Dfam$,
and it uses the algorithm for a $T$-intersecting uncrossable family to find a node set that covers all demand cuts in $\Ffam$.
After adding to the solution all the nodes in the obtained node set,
the algorithm updates $\Dfam$, setting it equal to the residual family, which consists of all the demand
cuts that are not
covered by the current solution.
This is repeated until $\Dfam$ becomes $T$-intersecting uncrossable.
Nutov proved that the algorithm for a $T$-intersecting uncrossable family is invoked
$O(k)$ times.
Below,
we present an algorithm for
covering a $T$-intersecting uncrossable family in Section~\ref{sssec.covering}, and
combines the Nutov's decomposition result and the covering algorithm in Section~\ref{ssec.combine}.

 \subsubsection{Covering a $T$-intersecting uncrossable family of demand cuts}
\label{sssec.covering}
 
 Here, we explain how to cover a $T$-intersecting uncrossable family $\Ffam$ of demand cuts.
 First, we introduce several properties of a $T$-intersecting uncrossable family.
 Let a \emph{min-core} signify a minimal demand cut in $\Ffam$.

 \begin{lemma}
  \label{lem.disjoint}
 Let $\Ffam$ be an uncrossable family of subsets of $V$. Let $X,Y \in
 \Ffam$.
 If $X$ is a min-core of $\Ffam$, then 
 either $X \subseteq Y$ or $X \cap Y^+=\emptyset=X^+\cap Y$ holds.
 In particular, if both $X$ and $Y$ are min-cores of $\Ffam$,
  the latter condition holds.
 \end{lemma}
  \begin{proof}
   Note that
   $X \cap Y^+=\emptyset=X^+\cap Y$ holds if and only if $X \cap Y^+ = \emptyset$.
 Suppose that $X \not\subseteq Y$ and
 $X \cap Y^+\neq \emptyset$.
 Since $\Ffam$ is uncrossable, $X \cap Y, X\cup Y \in \Ffam$ or
 $X\setminus Y^+, Y \setminus X^+ \in \Ffam$ holds.
 $X\not\subseteq Y$ implies $X \cap Y \subset X$,
  and $X\cap Y^+ \neq \emptyset$ implies $X \setminus Y^+ \subset X$.
  Hence, if either $X \cap Y \in \Ffam$
  or $X\setminus Y^+ \in \Ffam$ holds, we have a contradiction with the minimality of $X$.
  \end{proof}

For $S \subseteq V$,
let $\Ffam_S$ denote $\{X \in \Ffam \colon S \cap
\Gamma(X)=\emptyset\}$.

 \begin{lemma}
  \label{lem.residual}
 If $\Ffam \subseteq 2^V$ is $T$-intersecting uncrossable, then
  $\Ffam_S$ is also $T$-intersecting uncrossable
 for any $S \subseteq V$.
 \end{lemma}
\begin{proof}
 Let $X,Y \in \Ffam_S$. Since $\Ffam$ is uncrossable, $X\cap Y, X\cup Y
 \in \Ffam$ or $X\setminus Y^+, Y\setminus X^+ \in \Ffam$ holds.

 Suppose that the former holds.
 If $X \cap Y \not\in \Ffam_S$, then $\Gamma(X\cap Y)$ includes a node
 $v \in S$. Since $\Gamma(X\cap Y) \subseteq \Gamma(X)\cup \Gamma(Y)$,
 $v$ covers $X$ or $Y$. However, this contradicts the fact that $X,Y \in
 \Ffam_S$. Hence, $X\cap Y\in \Ffam_S$.
 Similarly, we can prove $X\cup Y \in \Ffam_S$, because
 $\Gamma(X\cup Y) \subseteq \Gamma(X)\cup \Gamma(Y)$.

 Next, suppose that the latter holds.
 We note that
 $\Gamma(X \setminus Y^+) \subseteq \Gamma(X) \cup \Gamma(Y)$ and
 $\Gamma(Y\setminus X^+) \subseteq \Gamma(X) \cup \Gamma(Y)$.
 Hence, as above, we can prove that
 $X \setminus Y^+, Y\setminus X^+ \in \Ffam_S$.
\end{proof}

Now we present  our algorithm for covering a $T$-intersecting uncrossable family $\Ffam$
  of demand cuts.
  The algorithm consists of the increase phase and the reverse deletion
 phase.
 First, we present the increase phase.
In this phase, the algorithm maintains the following variables.

  \begin{itemize}
   \item  The dual solution $y$.
	  This is initialized as $y(X):=0$, $X \in \Ffam$.
	  During the increase phase, $y$ is always a
	  feasible solution to \eqref{eq.dual-lp}.
   \item  A solution $S \subseteq V \setminus T$.
	  This is initialized to $S$:=$\emptyset$. The increase phase terminates when $S$ covers all the demand cuts in $\Ffam$.
  \end{itemize}

   In each iteration, the algorithm simultaneously increases $y(X)$ for each min-core
   $X$ of $\Ffam_S$.
   For ease of presentation, we will consider this over time.
   During $\epsilon$ units of time, $y(X)$ is increased by
   $\epsilon$ for each min-core $X$ of $\Ffam_S$.
   When $\sum_{X \in \Ffam\colon v \in \Gamma(X)}y(X)$ becomes equal to $w(v)$ for some $v
 \in V \setminus (T \cup S)$,
 the algorithm stops increasing $y$
 and adds $v$ to $S$.
After this update, if $S$ covers all demand cuts in $\Ffam$, then
the algorithm terminates the increase phase and proceeds to the
 reverse deletion phase.
 Otherwise, the algorithm proceeds to the next iteration of the increase
   phase.
 By definition, in these steps, $y$ is always a feasible solution to
 \eqref{eq.dual-lp}.
 At the end of the increase phase, $S$ covers all the demand cuts in $\Ffam$.
 
 Suppose that the increase phase
 starts at time $0$ and ends at time $\Delta$.
 Let $\tau \in [0,\Delta]$ be a moment during the increase phase.
 Let $S_{\tau}, \Ffam_{\tau}, \Mfam_{\tau}$ denote $S$, $\Ffam_S$,
 and the family of min-cores of $\Ffam_S$, respectively,
 at time $\tau$.
 For each $v \in V$, let $d_{\tau}(v)$ denote $|\{X \in \Mfam_{\tau}\colon v
 \in \Gamma(X)\}|$.
 Let $\Lfam =\bigcup_{\tau \in [0,\Delta]}\Mfam_{\tau}$.
 For the analysis given below, we observe the following properties.

  \begin{lemma}
   \label{lem.stronglaminar}
  $\Lfam$ is strongly laminar (i.e., any $X,Y\in \Lfam$ satisfy $X\subseteq Y$, $Y \subseteq X$, or $X \cap Y^+ = \emptyset = X^+ \cap Y$ holds).
  \end{lemma}
\begin{proof}
 Let $X,Y \in \Lfam$. Suppose that $X \in \Mfam_{\tau'}$ and $Y \in
 \Mfam_{\tau''}$
 for some $\tau' \leq \tau''$.
 Then, $Y \in \Ffam_{\tau'}$, because $Y \in \Mfam_{\tau''}$ means that $Y$ is
 not covered by $S_{\tau'}$.
 Hence, by Lemma~\ref{lem.disjoint}, $X \subseteq Y$ or $X\cap
 Y^+=\emptyset=X^+ \cap Y$ holds.
\end{proof}
 
  \begin{lemma}
   \label{lem.degree-udg}
   If $G$ is a unit disk graph, 
   $d_{\tau}(v)\leq 5$ holds for any $v \in V \setminus T$ and $\tau \in
   [0,\Delta]$.
  \end{lemma}
 \begin{proof}
  Let
  $X_1,\ldots,X_{d_{\tau}(v)}$ be the
  members of $\Mfam_{\tau}$ whose neighbor sets include
  $v$.
  Let $u_i$ denote
  a neighbor of $v$ in $X_i$ for each $i=1,\ldots,d_{\tau}(v)$.
  Notice that by Lemma~\ref{lem.residual}, $\Ffam_{\tau}$ is uncrossable.
  Thus, if $i\neq j$, then by Lemma~\ref{lem.disjoint}, $X_i \cap X_j^+ = \emptyset = X_i^+ \cap X_j$.
  If $d_{\tau}(v) \geq 6$,
  $u_i$ and $u_j$ are adjacent 
  for some $i,j \in \{1,\ldots,d_{\tau}(v)\}$ with $i\neq j$.
  However,
  this indicates that $u_j \in X^+_i$, which contradicts $X_j \cap
  X^+_i =\emptyset$.
  \end{proof}

  In the reverse deletion phase, the algorithm 
  modifies $S$ into an inclusionwise minimal node set that covers $\Ffam$ as follows.
  Let $S:=\{v_1,\ldots,v_{|S|}\}$, where $v_i$ is the $i$-th node added to $S$ 
  in the increase phase
  for each $i=1,\ldots,|S|$.
  The reverse deletion phase investigates nodes $v_i \in S$ in decreasing order
  of their subscripts. If
  $S \setminus \{v_i\}$ covers all demand cuts in $\Ffam$,
  $v_i$ is removed from $S$.
  Let $\tilde{S}$ denote $S$ after all nodes have been investigated.
  The algorithm outputs $\tilde{S}$ as a solution.
  We will show that $\tilde{S}$ is a 15-approximate solution.
  
  \begin{lemma}\label{lem.witness}
   Let $\tau \in [0,\Delta]$,
   and let $v_i \in \tilde{S}$ with $d_{\tau}(v_i)\geq 1$.
   There exist
   $W \in \Ffam_{\tau}$ and $X \in \Mfam_{\tau}$ such that $\Gamma(W) \cap (\tilde{S} \cup
   \{v_1,\ldots,v_i\})=\{v_i\}$,
    $v_i \in \Gamma(X)$, and
   $X \subseteq W$ or $X \cap W^+= \emptyset = X^+ \cap W$.
  \end{lemma}
   \begin{proof}
    When $v_i$ is investigated in the reverse deletion phase,
    the solution set $S$ is $\tilde{S}\cup \{v_1,\ldots,v_i\}$.
    Hence, 
    $v_{i}\in \tilde{S}$ implies that
    there exists $W \in \Ffam$
    such that 
    $\Gamma(W) \cap (\tilde{S} \cup \{v_1,\ldots,v_i\})=\{v_i\}$.
      Note that $W \in \Ffam_{\tau}$ holds because
    no node in $\{v_1,\ldots,v_{i-1}\}$ covers $W$.
    Since $d_{\tau}(v_i)\geq 1$,
     there exists $X \in \Mfam_{\tau}$ with $v_i \in \Gamma(X)$.
    By Lemma~\ref{lem.disjoint}, $X \subseteq W$ or $X \cap W^+=\emptyset=X^+ \cap W$ holds.
   \end{proof}

  For $v_i \in \tilde{S}$, we will call $(W,X)$ in Lemma~\ref{lem.witness} a
  \emph{witness pair} of $v_i$.

 \begin{figure}[t]
  \centering
 \includegraphics[bb=13 12 95 69]{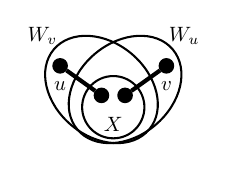}
 \caption{When $X=X_u=X_v$ for first-type nodes $u$ and $v$}
 \label{fig:counting2}
 \end{figure}

 \begin{figure}[t]
  \centering
\includegraphics[bb=9 10 161 82]{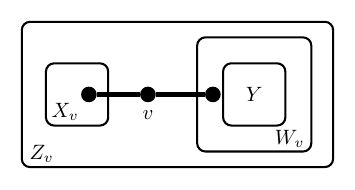}
 \caption{Relationship of $X_v$, $W_v$, $Y$, and $Z_v$ for a second-type node $v$}
 \label{fig:counting}
 \end{figure}
  
\begin{lemma}\label{lem.degree}
 $\sum_{v \in \tilde{S}} d_{\tau}(v) \leq 15|\Mfam_{\tau}|-5$ for
 any $\tau \in [0,\Delta]$.
\end{lemma}
  \begin{proof}
   Let $v$ be a node in $\tilde{S}$ such that 
   $d_{\tau}(v) \geq 1$, from which
   $v \not\in S_{\tau}$ follows
   because $S_{\tau}$ covers no demand cut in $\Mfam_{\tau}$ by
   their definitions.
   We categorize such a node $v$ into two types.
   If there exists a witness pair
   $(W_v,X_v)$ of $v$
   such that $X_v \subseteq W_v$, $v$ is said to be of the first type;
   otherwise, $v$ is said to be of the second type.

   Let us count the number of nodes $v$ of the first type.
   Let $X \in \Mfam_{\tau}$. Suppose that there are two nodes of the first type, $u,v \in
   \tilde{S}$,
   such that $X =X_u= X_v$.
   Then $X \subseteq W_u \cap W_v$ holds.
   Recall that $\Gamma(W_u)\cap \tilde{S}=\{u\}$, 
   $\Gamma(W_v)\cap \tilde{S}=\{v\}$,
   and $u,v \in \Gamma(X) \subseteq (W_u \cap W_v)^+$
   (see Figure~\ref{fig:counting2}).
   These imply that $\Gamma(W_u \cup W_v)\cap \tilde{S}=\emptyset$.
   Since $W_{u}\cap W_v \cap T \supseteq X \cap T \neq \emptyset$
   and $\Ffam$ is $T$-intersecting, we have
   $W_u \cup W_v \in \Ffam$. However, this contradicts the
   definition of $\tilde{S}$.
   Therefore, for each $X \in \Mfam_{\tau}$,
   there exists at most one node of the first type with $X_v=X$. That is to say,
  there are at most $|\Mfam_{\tau}|$ nodes of the first type.

  Next, we count the number of nodes of the second type.
   Let $v$ be a node of the second type.
   There exists $Y \in \Mfam_{\tau}$
   such that  $Y \subseteq W_v$, because $W_v$ is not covered by
   any node that was added to $S$ earlier than $v$ in the increase phase.
   Since $v$ is not the node of the first type,
   $v$ is not included in $\Gamma(Y)$, and thus $Y \neq X_v$.
   Let $\Lfam_{\tau}=\bigcup_{\tau' \in [\tau,\Delta]}\Mfam_{\tau'}$.
   Since
   $\Lfam_{\tau} \subseteq \Lfam$, and by Lemma~\ref{lem.stronglaminar},
   $\Lfam$ is a strongly laminar family,
   $\Lfam_{\tau}$ is a strongly laminar family.
   We assume that $\Lfam_{\tau}$ has a unique maximal member;
   if
   there is more than one maximal member of 
   $\Lfam_{\tau}$, we add the node set $V$ to $\Lfam_{\tau}$, which has
   no effect on the following discussion.
   Let $Z$ be a member of $\Lfam_{\tau}$ such that $Z$ became a
   min-core of the residual family of $\Ffam$ when $v$ was added to $S$
   in the increase phase.
   Then $v \in Z$, and hence $X_v \subseteq Z$ and $Y \subseteq
   W_v\subseteq Z$,
   because $v \in X^+_v$ and $v \in W^+_v$.
   Let $Z_v$ be a minimal member of $\Lfam_{\tau}$ such that $X_v
   \subseteq Z_v$ and $W_v \subseteq Z_v$;
   Figure~\ref{fig:counting} illustrates these definitions.
   Note that $Z_v$ has at least two children in $\Lfam_{\tau}$.
   Suppose that there exists a node of the second type, $u \in \tilde{S}\setminus \{v\}$,
   such that $d_{\tau}(u)\geq 1$ and   
   $W_v \subseteq W_u \subset Z_v$.
   Then,
   $v \in W_u$ holds, because $\Gamma(W_u) \cap \tilde{S}=\{u\}$,
   and $X_v \subseteq W_u$ because of the strong laminarity of
   $\Lfam_{\tau}$.
   Since this contradicts the
   definition of $Z_v$,
   no such $u$ exists.
   In summary, this means that
   the number of nodes of the second type
   is at most $\sum_{Z \in \Lfam_{\tau}: \text{ch}(Z)\geq 2} \text{ch}(Z)$,
   where $\text{ch}(Z)$ denotes the number of children of $Z$ in
   $\Lfam_{\tau}$.
   Since the leaf set of $\Lfam_{\tau}$ is $\Mfam_{\tau}$,
   we have $\sum_{Z \in \Lfam_{\tau}: \text{ch}(Z)\geq 2} \text{ch}(Z) \leq 2|\Mfam_{\tau}|-1$.

    Thus, $|\{v \in \tilde{S} \colon d_{\tau}(v) \geq 1\}| \leq  3|\Mfam_{\tau}|-1$,
   and $d_{\tau}(v) \leq 5$ for each $v \in \tilde{S}$, and the lemma has been proven.
  \end{proof} 

  \begin{theorem}
   \label{thm.coveringuncross}
  If $\Ffam$ is a $T$-intersecting uncrossable family of demand cuts on a unit disk graph,
  there exists a $15$-approximation algorithm for finding a minimum
  weight node set that covers all demand cuts in $\Ffam$.
  \end{theorem}
  \begin{proof}
  We will prove that the algorithm presented above is a $15$-approximation
   algorithm.
   By definition, the algorithm computes a feasible solution
   $\tilde{S}$ to the problem, and a solution $y \in \Rset_+^{\Ffam}$
  is a feasible solution to \eqref{eq.dual-lp}.
   Since $\sum_{X \in \Ffam}y(X)$ is a lower bound on the optimal value, it suffices to prove that $\sum_{v \in \tilde{S}}w(v)
   \leq 15\sum_{X \in \Ffam}y(X)$.
   
   When the increase phase terminates,
   $y$ satisfies
 $\sum_{X \in \Ffam}y(X)=
 \int_0^{\Delta} |\Mfam_{\tau}| d\tau$.
   Moreover,
   for each node $v \in V \setminus T$ and $\tau \in [0,\Delta]$,
   $\frac{d}{d\tau} \sum_{X \in \Ffam\colon v \in
   \Gamma(X)}y(X)=d_{\tau}(v)$ holds.
   For each $v \in \tilde{S}$,
   $w(v) = \int_{0}^{\Delta} d_{\tau}(v) d\tau$ holds,
   because $w(v)=\sum_{X \in \Ffam \colon v \in \Gamma(X)}y(X)$ holds
 when the algorithm terminates.
 By Lemma~\ref{lem.degree},
 $\sum_{v \in \tilde{S}}w(v)
 = \sum_{v \in \tilde{S}}\int_0^{\Delta} d_{\tau}(v)d\tau
 \leq 15 \int_0^{\Delta} |\Mfam_{\tau}| d\tau = 15\sum_{X \in \Ffam}y(X)$.
  \end{proof}

    Although we illustrated the algorithm by using a continuous measure of time,
  it can be easily discretized;
  Algorithm~\ref{alg.covering} 
  describes the details of our algorithm
  for covering a $T$-intersecting uncrossable family of demand cuts.

     \begin{algorithm}
  \caption{Covering algorithm for a $T$-intersecting uncrossable family of demand cuts}
  \label{alg.covering}
  \begin{algorithmic}
   \REQUIRE
   a unit disk graph $G=(V,E)$, $T  \subseteq V$,
   a nonnegative weight $w(v)$ of each node $v \in V \setminus T$,
   and a $T$-intersecting uncrossable family $\Ffam \subseteq \Dfam$ 
   \ENSURE $S \subseteq V \setminus T$ that covers all demand cuts in $\Ffam$
   \STATE  $S \longleftarrow \emptyset$, $i \leftarrow 0$
   \STATE $\overline{w}(v) \longleftarrow w(v)$ for each $v \in V
   \setminus T$
   \WHILE{$\Ffam_S:= \{X \in \Ffam \colon S \cap \Gamma(X)=\emptyset\} \neq \emptyset$}
   \STATE $i \longleftarrow i+1$
   \STATE $\Mfam \leftarrow \{\text{min-cores of } \Ffam_S\}$
   \STATE $d(v) \longleftarrow |\{X \in \Mfam \colon v \in
   \Gamma(X)\}|$ for each $v \in V \setminus (T \cup S)$
   \STATE
   $\alpha \leftarrow \min_{v \in V \setminus (T \cup S)}
   \overline{w}(v)/d(v)$
   \STATE  $v_i \leftarrow \arg \min_{v \in V \setminus (T \cup S)}
   \overline{w}(v)/d(v)$
   \STATE $S \longleftarrow S \cup \{v_i\}$
   \STATE $\overline{w}(v) \longleftarrow \overline{w}(v) - \alpha d(v)$
   for each $v \in V\setminus (T \cup S)$
   \ENDWHILE
   \FOR{$j = i-1, \ldots,1$}
   \STATE if $S \setminus \{v_j\}$ covers $\Ffam$, then $S \leftarrow S \setminus \{v_j\}$
   \ENDFOR
   \STATE output $S$
  \end{algorithmic}
   \end{algorithm}

Let us discuss the running time of Algorithm~\ref{alg.covering}.
Lemma~\ref{lem.disjoint} indicates that the number of
min-cores in $\Ffam_S$ is $O(n)$ for any $S$.
Hence the algorithm runs in polynomial time if all min-cores of
$\Ffam_S$ can be computed in polynomial time.
If $\Ffam$ arises from the demand cuts of the augmentation problem in
the decomposition of Nutov~\cite{Nutov12},
a min-core of 
$\Ffam_S$ corresponds to a minimal node cut in the graph $G[T\cup S]$.
Thus all min-cores can be computed by applying a max-flow algorithm
repeatedly in this case.
  
  \subsubsection{Combined decomposition and covering algorithm}
  \label{ssec.combine}
  
  We now summarize our algorithm for the augmentation problem.
  As mentioned earlier, 
  due to Nutov~\cite{Nutov12},
  we can find a node set that covers all demand cuts
  by applying an
  algorithm for covering a $T$-uncrossable family of demand cuts $O(k)$
  times.
  Since Theorem~\ref{thm.coveringuncross} gives a constant-approximation algorithm
  for covering a $T$-intersecting uncrossable family,
  we have an $O(k)$-approximation algorithm for covering all demand
  cuts.
  Recall that a demand cut is defined as a Steiner $(T,r)$-cut for a
  fixed $r \in T$.
  For covering all Steiner $T$-cuts $X$ with $|\Gamma_T(X)|=k-1$,
  we choose arbitrary $k$ nodes $r_1,\ldots,r_k \in T$,
  and apply the algorithm fixing $r$ to each of $r_1,\ldots,r_k$.
  The union of the obtained solutions is
  a feasible solution for the augmentation problem
  because, for any Steiner $T$-cut $X$,
  there exists a root $r_i$ such that $X$ or $V \setminus X^+$ is a
  Steiner $(T,r_i)$-cut.
  This presents 
  an $O(k^2)$-approximate
  solution for the augmentation problem.
  Therefore, we arrive at the following conclusion.

   \begin{corollary}
    \label{cor.augmentation}
    The augmentation problem admits 
    an $O(k^2)$-approximation algorithm.
    It outputs a solution $S$ such that $\sum_{v \in S}w(v)$
   is at most $O(k^2)$ times the optimal value of \eqref{eq.primal-lp}.
   \end{corollary}

 For the weighted $(k,m)$-CDS problem,
 we first compute an $m$-dominating set $T$, using the algorithm given in
 Section~\ref{subs.dominating}.
 We incrementally increase the connectivity of $T$ 
 by solving the augmentation problem.
 This obviously gives an $O(k^3)$-approximation algorithm.
 This approximation factor can be slightly improved, as follows.

 \begin{corollary}\label{cor.primal-dual}
  There exists an $O(k^2 \log k)$-approximation algorithm for the
  weighted $(k,m)$-CDS problem.
 \end{corollary}
 \begin{proof}
  Let $S^*$ denote an optimal solution for the weighted $(k,m)$-CDS
  problem, and let $x^* \in \{0,1\}^V$ denote its characteristic vector
  (i.e., $x^*(v)=1$ if $v \in S^*$, and $x^*(v)=0$ if $v \in V \setminus
  S^*$).
  For each $S \subseteq V$, we abbreviate $\sum_{v \in S}w(v)$ to $w(S)$.

 We will show that the algorithm described above achieves the approximation
  factor $O(k^2 \log k)$.
  Recall that we use the algorithm given in
  Corollary~\ref{cor.augmentation} for the augmentation
  problem, but with the connectivity requirement $k'$ changed
  from $1$ to $k$.
  Let $S_{k'}$ denote the solution output by the algorithm for the
  augmentation problem when the connectivity requirement is $k'$.
  Note that in this case, the node set $T$ in the input is $\bigcup_{i=0}^{k'-1}S_{i}$, where $S_0$ is the
  $m$-dominating set computed by the algorithm given in Section~\ref{subs.dominating}.
 Note that $S^* \cup \left(\bigcup_{i=0}^{k'-1}S_{i}\right)$ is $k$-connected.
  Hence, when $T=\bigcup_{i=0}^{k'-1}S_{i}$, for each demand cut $X$,
  $\Gamma(X)\setminus T$ includes at least $k-k'+1$ nodes in $S^*$.
  This implies that $x^*/(k-k'+1)$ is a feasible solution to \eqref{eq.primal-lp}
  when the connectivity requirement is $k'$,
  and thus $w(S_{k'}) = O(k^2) \cdot w(S^*)/ (k-k'+1)$.
  The weight of the $(k,m)$-CDS output by our algorithm
  is at most
  $\sum_{k'=0}^k w(S_{k'})
   = O(k^2) \cdot w(S^*) \sum_{k'=0}^k \frac{1}{k-k'+1}
   = O(k^2 \log k) \cdot w(S^*).$
 \end{proof}

 \section{Conclusion}
 \label{sec.conclusion}

 We presented two constant-approximation algorithms for the
 unweighted $(k,m)$-CDS problem
 and for the weighted $(k,m)$-CDS problem in unit
 disk graphs.
 The first of these is a simple algorithm that can be applied to a fairly general class of graphs,
 although it is restricted to the unweighted
 $(k,m)$-CDS problem.
 The second is a primal-dual algorithm that has a better approximation factor and
 can be applied to the weighted $(k,m)$-CDS problem.
 Both algorithms need an assumption that $k \leq m$.

 Besides the $m$-dominating sets,
  there are many other variations of dominating sets in graphs. For example,
  a subset $S$ of a node set $V$ is called an \emph{$m$-tuple dominating set}
  if each node in the graph (including those in $S$)
  has $m$ neighbors in $S$,
  and it is called a \emph{vector dominating set}
  if each node $v$ outside of $S$ has $d(v)$ nodes in $S$
  for a given $|V|$-dimensional vector $d$.
  Refer to~\cite{CicaleseMV13} for other variations.
  Our algorithms for the augmentation problem can be used for
  increasing the connectivity of 
  these  variations if each node outside the solution has $k$
  neighbors
  in the solution, where $k$ is the required connectivity.

 Our primal-dual algorithm
 for the weighted problem
 requires the ellipsoid method
 for computing an $m$-dominating set when $k+1 < m$ (when $m \leq k+1$,
 using the ellipsoid method can be avoided; see the last paragraph of Section~\ref{subs.dominating}).
 However, the ellipsoid method is not practical, so 
another interesting future work will be to invent a
 constant-approximation algorithm for the minimum weight $m$-dominating
 set problem, that does not rely on the ellipsoid method.

\section*{Acknowledgements}
The author thanks anonymous referees for their careful reading and many useful comments.
This work was supported by JST ERATO Grant Number JPMJER1201, and
JSPS KAKENHI Grant Number 17K00040.

\appendix
\section{Improved analysis of Shang et al. for the unweighted $(2,m)$-CDS}
\label{app.rect}

Shang~et~al.\ \cite{Shang:2007jg} gave an
approximation
algorithm for the unweighted $(2,m)$-CDS problem.
They claimed that their algorithm achieves an approximation factor of
$5+25/m$ for $2 \leq m \leq 5$, and $11$ for $m \geq 6$.
However, Shi~et~al.~\cite{ShiZZW16} pointed out that its analysis
contains a flaw. Shi~et~al.~\cite{ShiZZW16} also rectified the analysis, and showed
that
its approximation factor is $15+15/m$ for $2 \leq m \leq 5$
and $21$ for $m \geq 6$.
Simultaneously, 
Shi~et~al.~\cite{ShiZZW16} presented an algorithm for the unweighted $(2,m)$-CDS
problem on general graphs, and proved that their algorithm achieves
approximation factors $7+5/m+2\ln(5+5/m)$ for $2 \leq m\leq 5$,
and 11 for $m \geq 6$ if the graphs are restricted to unit disk graphs.
In this section, we present an improved analysis of  Shang~et~al.~\cite{Shang:2007jg}.
It gives approximation factors 
  $5+35/m$ for $2 \leq m \leq 5$, and 
  $13-5/m$ for $m\geq 6$.
Although these are not better than those of 
Shi~et~al.~\cite{ShiZZW16}, we believe that it is worth noting.

Let us begin with illustrating the algorithm of Shang~et~al.~\cite{Shang:2007jg}.
Let $\opt$ denote the minimum size of $(2,m)$-CDSs.
Let $I_{i}$ be a maximal independent set of $G[V \setminus
\bigcup_{i'=0}^{i-1} I_{i'}]$ for each $i =1,\ldots,m$, where $I_0=\emptyset$.
The algorithm first computes $I_1,\ldots,I_m$,
and $C \subseteq V \setminus I_1$ such that $|C|\leq |I_1|$ and
$G[C \cup I_1]$ is connected.
The following properties are proven.
 \begin{itemize}
  \item[(i)] $|I_i| \leq \max\{5/m,1\}\opt$ for each $i=1,\ldots,m$.
  \item[(ii)] $\bigcup_{i'=1}^{i}I_{i'}$ is an $i$-dominating set for
	      each $i=1,\ldots,m$,
	      and hence $T:= \bigcup_{i=1}^m I_i \cup C$ is a $(1,m)$-CDS.
  \item[(iii)] Each cut-node of $G[T]$ is included in $I_1$ or in $C$.
  \item[(iv)] $|T| \leq (5+5/m)\opt$ for $m\leq 5$,
	      and $|T| \leq (7-5/m)\opt$ for $m \geq 6$
	      (this is slightly better than the conclusion in
	      \cite{Shang:2007jg}, but they proved this).
 \end{itemize}

 After this step,
the algorithm computes a
node set $S$ such that $T\cup S$ is 2-connected as follows.
First, $S$ is initialized to be an empty set.
Then, the algorithm computes
a $T$-cut $X$ such that $|\Gamma_{T}(X)|=1$
and no $T$-path in $G[T \cup S]$
covers $X$.
For this $T$-cut $X$,
the algorithm finds a $T$-path that covers $X$ with at most two inner nodes,
and it adds these inner nodes to the solution $S$.
This procedure is repeated until $T \cup S$ becomes 2-connected,
and the algorithm outputs $T\cup S$
as a $(2,m)$-CDS.

Shang~et~al.\ \cite{Shang:2007jg} claimed that the number of iterations is at most
the number of cut-nodes in $G[T]$, which is 
at most $|I_1| + |C| \leq 2|I_1| \leq 2\max\{5/m,1\}\opt$, due to
properties (i) and (iii).
Since each iteration adds at most two nodes, when the algorithm terminates, 
 $|S| \leq 4\max\{5/m,1\}\opt$.
 This is their analysis of the algorithm.

 However, the number of iterations is not bounded by the number of cut-nodes in $G[T]$.
 This can be observed by considering a star with $n$ leaves.
 The star has only one cut-node. To make it 2-connected,
 we need to add $n-1$ paths to connect the leaves.

 We claim that a correct upper-bound on the number of iterations
 is the number of blocks of $G[T]$,
 and this number is bounded by $|I_1|+|C|+|I_2| -1 \leq 3\max\{5/m,1\}\opt-1$.
 A block of a connected graph is a maximal set of nodes that induces a connected subgraph without cut-nodes.
 Each block consists of at least two nodes.
 If a block consists of more than nodes, then it is a 2-connected component of the graph.
 If a block consists of only two nodes, then it induces an edge.

 For proving our claim, 
  let us consider the tree $F$ that represents the block
  decomposition of $G[T]$.
  Namely, the node set of $F$ is the disjoint union of two node sets $B$
  and $W$. Each node in $B$ corresponds to a block of
  $G[T]$,
  and each node in $W$ corresponds to a cut-node of $G[T]$.
  In what follows, we identify a node in $B$ with the corresponding
  block of $G[T]$, and we identify a node in $W$ with the
  corresponding
  cut-node of $G[T]$.
  A node $b \in B$ and a node $w \in W$ are joined by an edge in $F$
  when the block $b$ includes the cut-node $w$.

  In each iteration of the algorithm, a $T$-path is selected to
  join two different blocks of $G[T]$,
  and the inner nodes of the path are
  added to the solution $S$.
  Let $u$ and $v$ denote the end nodes of a $T$-path $P$, and 
  let $\rho_u$ be a block that includes $u$.
  If more than one block includes $u$, we let $\rho_u$
  denote the one nearest to the blocks including $v$ on $F$;
  $\rho_v$ is defined in the same way.
  Let $x$ be a cut-node of $G[T]$,
  and let $b$ and $b'$ be blocks that include $x$.
  When the algorithm chooses a $T$-path $P$,
  we add a virtual edge 
  that joins
  $b$ and $b'$ 
  if $x$, $b$, and $b'$ are on the path between
  $\rho_u$ and $\rho_v$ on $F$.
  
  \begin{lemma}
   $T\cup S$ is $2$-connected if
   virtual edges induce a connected graph
   on the set of neighbors of
    each cut-node $x$ in $F$.
  \end{lemma}
  \begin{proof}
   Suppose that $T\cup S$ is not 2-connected even if the condition holds. Then there 
   exists a cut-node $x$ of $G[T\cup S]$.
   In other words, after removing $x$ from $G[T\cup S]$, some neighbors $y$ and $y'$ of $x$
   are included in the different connected component.
   Since $T$ is $m$-dominating and $m \geq 2$, we can assume that $x,y,y' \in T$.
   Hence $x$ is also a cut-node of $G[T]$.
      
   Let $b$ and $b'$ be the blocks of $G[T]$ including $y$ and $y'$, respectively. 
   Since $y$ and $y'$ are neighbors of $x$, both $b$ and $b'$ include $x$ (i.e., 
   $b$ and $b'$ are neighbors of $x$ in $F$).
   By the condition of the lemma,
   $b$ and $b'$ are connected by a path of the virtual edges. 
   This implies that there exists a path on $G[T\cup S]$ that connects $y$ and $y'$,
   and that does not pass through $x$.
   Hence, $x$ is not a cut-node in $G[T\cup S]$, which is a contradiction.
  \end{proof}

  The following lemma presents a bound on the number of iterations in this algorithm.
 
 \begin{lemma}\label{lem.iterationfor2}
  The number of iterations is at most
  $3\max\{5/m,1\}\opt-1$.
 \end{lemma}
   \begin{proof}
    Let $x$ be a cut-node on $F$,
    and let $\psi_x$ denote the 
    number of connected components 
    induced by the virtual edges
    on the neighbor set of $x$.
    In each iteration of the algorithm,
    $\psi_x$ is decreased by at least one
    for some cut-node $x$.
    When the first iteration begins,
    $\psi_x$ is equal to the degree of $x$ in $F$.
    Since all leaves in $F$ are included in $B$,
    $\sum_{x \in W} \psi_x=|B|+|W|$ holds at the beginning of the first iteration.
    The iterations terminate when
    $\sum_{x \in W} \psi_x=|W|$.
    Hence,
    the number of iterations is at most $|B|$.
    We will determine $|B|$ below.

    Let $H$ be a spanning tree on $G[I_1 \cup C]$.
    We show that each block contains a node in
    $I_2\setminus C$
    or an edge in $H$.
    To see this, suppose that a block $b$ of $G[T]$ contains 
    no edge in $H$.
    If $b$ contains more than one node in $I_1 \cup C$,
    then it includes an edge in $H$.
    Since, by the assumption, this does not happen, $b$ includes
    at most one node in $I_1 \cup C$.
    Hence, there exists a node 
    $v \in \bigcup_{i=2}^m I_i \setminus C$ in $b$.
    If $v \in I_2$, we are done. Suppose that $v \not\in I_2$.
    By property (ii),
    $v$ has neighbors in $I_1$ and in $I_2$.
    By property (iii), $v$ is not a cut-node in $G[T]$, so
    these neighbors must be inside $b$.
    If the neighbor in $I_2$ is included in $C$,
    $b$ contains two nodes in $I_1 \cup C$.
    Since this contradicts the assumption, $b$
    contains a node in $I_2\setminus C$.

    Nodes in $\bigcup_{i=2}^m I_2 \setminus C$ and edges in $G[T]$ are
    not contained in more than one block of $G[T]$.
   The number of edges in $H$ is $|I_1| + |C|-1 \leq
   2\max\{5/m,1\}\opt-1$,
    and $|I_2| \leq \max\{5/m,1\}\opt$ by property (i).
    Hence $|B| \leq 3\max\{5/m,1\}\opt-1$.
   \end{proof} 

  From Lemma~\ref{lem.iterationfor2} and property (iv), we obtain the following
  approximation guarantee for the algorithm.
  \begin{theorem}
   The algorithm of Shang~et~al.\ \cite{Shang:2007jg}
   attains an approximation factor
   $5+35/m$ for $2 \leq m \leq 5$, and 
   $13-5/m$ for $m\geq 6$.
  \end{theorem}

\end{document}